\documentclass{amsart}
\usepackage{amsmath}
\usepackage{amssymb,amsthm}
\usepackage{graphicx}
\usepackage{tikz}
\numberwithin{equation}{section}
\newtheorem{thm}{Theorem}[section]
\newtheorem{rem}[thm]{Remark}

\newtheorem{lem}[thm]{Lemma}

\title[electronic densities in molecules]{Upper bounds of local electronic densities in molecules}
\author{Sohei Ashida}
\begin{document}
\maketitle

\begin{abstract}
The eigenfunctions of electronic Hamiltonians determine the stable structures and dynamics of molecules through the local distributions of their densities. In this paper an a priori upper bound for such local distributions of the densities is given. The bound means that concentration of electrons is prohibited due to the repulsion between the electrons. A relation between one-electron and two-electron densities resulting from the antisymmetry of the eigenfunctions plays a crucial role in the proof.
\end{abstract}

\section{Introduction and statement of the result}\label{firstsec}
As is well known the properties of molecules such as molecular structures and dynamics are determined by the dependence of the energy levels (eigenvalues of electronic Hamiltonians) on the nuclear positions. The electronic Hamiltonian of $N$ electrons and $L$ nuclei acting on $L^2(\mathbb R^{3N})$ is given by
\begin{equation}\label{myeq1.0.0}
H=H(\mathbf R):=\sum_{i=1}^Nh_i+\sum_{1\leq i<j\leq N}\frac{1}{|x_i-x_j|},
\end{equation}
with
\begin{equation}\label{myeq1.0}
h_i:=-\Delta_{x_i}-\sum_{l=1}^L\frac{Z_l}{|x_i-R_l|},
\end{equation}
where $x_i\in\mathbb R^3,\ i=1,\dots, N$ and $R_l\in\mathbb R^3,\ l=1,\dots L$ denote the positions of the electrons and the nuclei respectively, and $Z_l\in\mathbb N$ denotes the atomic number of the $l$th nucleus and $\mathbf R:=(R_1,\dots,R_L)$. The nuclear positions $R_l$ appear as parameters in $H(\mathbf R)$. Thus an eigenvalue $E=E(\mathbf R)$ of $H(\mathbf R)$ and the associated eigenfunction $\Psi=\Psi_{\mathbf R}$ also depend on $R_l$ as parameters. Since the wave function of electrons is known to be antisymmetric in physics, we consider only antisymmetric functions as eigenfunctions $\Psi$. Here antisymmetry means that $\Psi(x_{\sigma(1)},\dots,x_{\sigma(N)})=(\mathrm{sgn}\, \sigma)\Psi(x_1,\dots,x_N)$ for any element $\sigma\in S_N$ of the symmetric group $S_N$. We also assume that $\Psi$ is normalized i.e. $\lVert \Psi\rVert=1$ hereafter.

As is well-known, in the study of the dependence of $E(\mathbf R)$ on $\mathbf R$ to investigate the properties of molecules the electronic density appears naturally as follows. If $E(\mathbf R)$ is an eigenvalue of $H(\mathbf R)$ depending on $\mathbf R$ continuously and separated by a gap from the rest of the spectrum of $H(\mathbf R)$ in a domain $D\subset\mathbb R^{3L}$ of $\mathbf R$, the orthogonal projection $\Pi(\mathbf R)$ on the associated eigenspace is also continuous. Moreover, if $D$ is contractable, by the triviality of the bundle $\mathrm{Ran}\, \Pi(\mathbf R)$ (cf. \cite{BT,KMSW}) we can choose the eigenfunction $\Psi_{\mathbf R}$ associated with $E(\mathbf R)$ which depends on $\mathbf R$ continuously. The sum $U(\mathbf R):=E(\mathbf R)+\sum_{1\leq l<m\leq L}\frac{Z_lZ_m}{|R_l-R_m|}$ of the eigenvalue and the nuclear repulsion potential acts as an effective potential for the nuclei. By the well-known Hellmann-Feynman theorem $E(\mathbf R)$ is differentiable and the gradient of $U(\mathbf R)$ is given by
$$\nabla_{R_l}U(\mathbf R)=\langle\Psi_{\mathbf R},(\nabla_{R_l}W)\Psi_{\mathbf R}\rangle+\sum_{m\neq l}\frac{Z_lZ_m(R_m-R_l)}{|R_m-R_l|^3},$$
where $W:=-Z_l\sum_{i=1}^N\frac{1}{|x_i-R_l|}$. The first term in the right-hand side is expressed using the (one-electron) density
$$\rho_{\Psi}(x):=N\int_{\mathbb R^{3(N-1)}}|\Psi(x,x_2,\dots,x_N)|^2dx_2\dotsm dx_N,$$
as 
\begin{equation}\label{myeq1.1}
\langle\Psi_{\mathbf R},(\nabla_{R_l}W)\Psi_{\mathbf R}\rangle=-\int \frac{Z_l(x-R_l)}{|x-R_l|^3}\rho_{\Psi_{\mathbf R}}(x)dx.
\end{equation}
Here we used
\begin{align*}
&\int_{\mathbb R^{3(N-1)}}|\Psi(x_1,\dots,x,\dots,x_N)|^2dx_1\dotsm \widehat{dx_i}\dotsm dx_N\\
&\quad=\int_{\mathbb R^{3(N-1)}}|\Psi(x,x_2,\dots,x_N)|^2dx_2\dotsm dx_N,
\end{align*}
for $1\leq i\leq N$, which follows from the antisymmetry of $\Psi$, where $x$ appears as $i$th variable in the left-hand side and $\widehat{dx_i}$ means that $dx_i$ is excluded. (In some papers the density is defined by $\rho_{\Psi}(x):=\int_{\mathbb R^{3(N-1)}}|\Psi(x,x_2,\dots,x_N)|^2dx_2\dotsm dx_N$, and we have $\int_{\mathbb R^3}\rho_{\Psi}(x)dx =1$ in that case, but we have $\int_{\mathbb R^3}\rho_{\Psi}(x)dx =N$ when $\rho_{\Psi}(x)$ is defined as above.) Thus the state of the nuclei is governed by the density. In particular, the equilibrium structure of a molecule is determined by the condition $-\int \frac{Z_l(x-R_l)}{|x-R_l|^3}\rho_{\Psi_{\mathbf R}}(x)dx+\sum_{m\neq l}\frac{Z_lZ_m(R_m-R_l)}{|R_m-R_l|^3}=0$. Thus the electronic density is a fundamental quantity for the study of molecular structures and dynamics. From \eqref{myeq1.1} we can see that the local distribution of the density particularly near the nuclei is important. The physical meaning of the density is that the expectation value of the number of electrons found in a Lebesgue measurable set $\Omega$ is equal to $\int_{\Omega}\rho_{\Psi}(x)dx$. This is not obvious from the definition. We prove this fact in the appendix.

Regarding the local properties of the density proved in mathematically rigorous way, we have real analyticity of the density except at the positions $R_1,\dots,R_L$ of the nuclei (see \cite{FHS}) and the cusp condition which means that the density $\rho_{\Psi}(x)$ behaves as $e^{-Z_l|x-R_l|}\mu(x)$ near $R_l$, where $\mu$ is differentiable and the derivatives of $\mu$ are Lipschitz continuous (see \cite{FHS2}). For atoms i.e. $L=1$, it is also proved that the spherically averaged density of the ground state around the nucleus is positive for all radii (see \cite{FHS3}), and the spherically averaged density of any exponentially decaying eigenfunction is a $C^3$ function with estimates on the second and the third derivatives at the nuclear position $R_1$ in the cases of the eigenvalues smaller than the infimum of the essential spectrum of $H$ (see \cite{HS, FHS4}).

In this paper we prove an a priori upper bound of local distributions of the electronic density for molecules. More precisely we prove the following theorem.
\begin{thm}\label{ubd}
Assume that there exists a constant $a>0$ such that $\min_{l,m}|R_l-R_m|>a$. Let $\Psi$ be a normalized antisymmetric eigenfunction associated with the eigenvalue $E\leq 0$ of $H$. Then for any bounded Lebesgue measurable set $\Omega\subset \mathbb R^3$ we have
\begin{equation}\label{myeq1.2}
\int_{\Omega}\rho_{\Psi}(x)dx\leq N\left(\frac{4d_{\Omega}}{N-1}\{(8\sqrt2+6L^{2/3})\mathcal Za^{-1}+\mathcal Z^2\}+\frac{1}{4N^{2}}\right)^{1/2}+\frac{1}{2},
\end{equation}
where $d_{\Omega}:=\sup_{x,y\in\Omega}|x-y|$ and $\mathcal Z:=\max_lZ_l$.
\end{thm}
\begin{rem}
From the condition $E\leq 0$ in Theorem \ref{ubd} it is a natural question whether $H$ has positive eigenvalues or not. For $N$-body Hamiltonians in the configuration space with the variable of center of mass removed, the absence of positive eigenvalues has been proved by \cite{FH}. The absence of positive eigenvalues for the operators $H$ of the form in \eqref{myeq1.0.0} on the Hilbert space $L^2(\Omega_1^N)$ has been proved by \cite{IS}, where $\Omega_1\subset \mathbb R^3$ is a connected set such that $R_l\notin \Omega_1,\ l=1,\dots,L$. For $H$ in \eqref{myeq1.0.0} on $L^2(\mathbb R^{3N})$ the absence of positive eigenvalues is an open problem as far as the author knows.
\end{rem}
This result implies that the concentration of electrons in $\Omega$ is prohibited because the right-hand side has order $O(N^{1/2}L^{1/3})$ as $N$ and $L$ increase. Thus in particular if $L$ is proportional to $N$, the order is $O(N^{5/6})$. Note that if there were no repulsive interaction between the electrons, there could be a state $\Psi$ in which $N$ electrons are localized in the same bounded set $\Omega$ and $\rho_{\Psi}(x)$ could satisfy $\int_{\Omega}\rho_{\Psi}(x)dx=O(N)$ in such a case.

In order to understand the result more clearly let us introduce the notion of the probability measure associated with the density per electron:
\begin{equation}\label{myeq1.3}
P_{\Psi}(\Omega):=\int_{x_1\in\Omega}\int_{\mathbb R^{3(N-1)}}|\Psi(x_1,\dots,x_N)|^2dx_1\dotsm dx_N=\frac{1}{N}\int_{\Omega}\rho_{\Psi}(x)dx.
\end{equation}
Here $\int_{\mathbb R^{3(N-1)}}$ means that the integration with respect to the remaining variables $(x_2,\dots,x_N)$ in $\mathbb R^{3(N-1)}$.
$P_{\Psi}(\Omega)$ is the probability to find a distinguished electron in $\Omega$, although electrons are not distinguished physically. In terms of $P_{\Psi}(\Omega)$ the Theorem \ref{ubd} is written as
$$P_{\Psi}(\Omega)\leq\left(\frac{4d_{\Omega}}{N-1}\{(8\sqrt2+6L^{2/3})\mathcal Za^{-1}+\mathcal Z^2\}+\frac{1}{4N^{2}}\right)^{1/2}+\frac{1}{2N},$$
which implies that the probability to find a distinguished particle in $\Omega$ decreases with order $O(N^{-1/2}L^{1/3})$. In particular, if $L$ is proportional to $N$, the order is $O(N^{-1/6})$. 

It might be unexpected that the right-hand side of \eqref{myeq1.2} has order $O(d_{\Omega}^{1/2})$ as $d_{\Omega}$ increases. For if $\Omega$ is a ball, the Lebesgue measure $\mu(\Omega)$ of $\Omega$ increases as $\mu(\Omega)=O(d_{\Omega}^3)$, and thus we would expect $\int_{\Omega}\rho_{\Psi}(x)dx=O(d_{\Omega}^3)$ for moderate values of $d_{\Omega}$ if there were not repulsive interactions between the electrons. In reality, by \eqref{myeq1.2} the upper bound of the mean $\frac{1}{\mu(\Omega)}\int_{\Omega}\rho_{\Psi}(x)dx$ of the density in $\Omega$ decays with the order $O(d_{\Omega}^{-5/2})$. This would be a result of the repulsive interactions. In the proof of Theorem \ref{ubd} the repulsive energy $\langle \Psi,\sum_{1\leq i<j\leq N}\frac{1}{|x_i-x_j|}\Psi\rangle$ is bounded from above by a constant $C$ depending only on $N, L, \mathcal Z, a$, and roughly speaking this energy is bounded from below by $\left(\int_{\Omega}\rho_{\Psi}(x)dx\right)^2/d_{\Omega}$. Therefore, at the expense of rigor we have the estimate as $\left(\int_{\Omega}\rho_{\Psi}(x)dx\right)^2/d_{\Omega}\leq C$, from which the order $O(d_{\Omega}^{1/2})$ follows.

For the proof of Theorem \ref{ubd} we need to introduce another kind of density called two-electron density. The two-electron density $\nu_{\Psi}(x,y)$ is defined by
\begin{equation}\label{myeq1.4}
\nu_{\Psi}(x,y):=\int_{\mathbb R^{3(N-2)}}|\Psi(x,y,x_3,\dots,x_N)|^2dx_3\dotsm dx_N.
\end{equation}
We can see that by the antisymmetry of $\Psi$ we have
$$\int_{\mathbb R^{3(N-2)}}|\Psi(x_1,\dots,x,\dots,y,\dots,x_N)|^2dx_1\dotsm \widehat{dx_i}\dotsm \widehat{dx_j}\dotsm dx_N=\nu_{\Psi}(x,y),$$
for $1\leq i<j\leq N$, where $x$ and $y$ are $i$th and $j$th variables in the left-hand side. For the proof of Theorem \ref{ubd} we use an inequality relating the one-electron  and two-electron densities. The inequality plays a crucial role in the proof and seems fundamental, because we only need to assume that $\Psi$ is a normalized antisymmetric function and the assumption that $\Psi$ is an eigenfunction is not needed at all for the inequality. The antisymmetry is a strong condition  for functions of many variables. It ensures that the function is expanded as an infinite sum of the Slater determinants, and thus there exists a relation between the one-electron and two-electron densities. For the proof of the inequality we construct an orthonormal basis for $L^2(\mathbb R^3)$ in which each function has nonzero values only in $\Omega$ or $\Omega^c$. The purpose of the construction is to make integrals on $\Omega$ of products of two functions in the basis vanish unless the two functions are the same function, which is a property similar to that of the integral on the whole $\mathbb R^3$. 

We also need a lower bound to the spectrum of the one-electron Hamiltonian $h_i$ for the proof of Theorem \ref{ubd}. We obtain a bound by estimating the attractive potential between the electron and the nuclei. For the estimate near the nuclei we use the Hardy inequality for functions with compact supports. On the other region of $\mathbb R^3$ we estimate the potential converting it to electrostatic potentials by charge densities and relocating the charge densities.
Combining the lower bound obtained in this way and the relation between the one-electron and two-electron densities as above, Theorem \ref{ubd} is proved.

The organisation of the paper is as follows. In Section \ref{thirdsec} we prove a lemma about the infimum of the spectrum of one-electron Hamiltonians. In Section \ref{fourthsec} we explain the relation between the functions of all electrons and the tensor product of one-electron functions, and a lemma concerned with the relation between the one-electron and two-electron densities is given. Theorem \ref{ubd} is proved in Section \ref{fifthsec}. In the appendix we prove that the integral $\int_{\Omega}\rho_{\Psi}(x)dx$ of the one-electron density is the expectation value of the number of electrons found in $\Omega$.

\section{spectrum of one-electron Hamiltonian}\label{thirdsec}
The operators $h_i$ in \eqref{myeq1.0} are the same operator on $L^2(\mathbb R^3)$ with respect to the different variables $x_i$. The form of the operator is written as
$$h:=-\Delta_{x}-\sum_{l=1}^L\frac{Z_l}{|x-R_l|}.$$
In order to obtain an upper bound of the repulsive energy of the eigenfunction $\Psi$ we need a lower bound of the spectrum of $h$. First we prove an estimate for the attractive potentials.
\begin{lem}\label{formb}
Assume that there exists a constant $b>0$ such that $\min_{l,m}|R_l-R_m|> b$. Then for any $f\in H^1(\mathbb R^3)$ we have
\begin{equation}\label{myeq2.1}
\left\langle f,\sum_{l=1}^L\frac{1}{|x-R_l|}f\right\rangle\leq \lVert \nabla f\rVert^2 + \{(16\sqrt2+12L^{2/3})b^{-1}+2\}\lVert f\rVert^2.
\end{equation}
\end{lem}
\begin{proof}
Let $\chi(r)\in C^{\infty}(\mathbb R)$ be a function such that
$$\chi(r)=\begin{cases}
1 &r\leq b/4\\
0 &r\geq b/2
\end{cases},$$
$0\leq\chi(r)\leq 1$ and $|\chi'(r)|\leq 8/b$. We can easily see that such a function exists. We set $\chi_l(x):=\chi(|x-R_l|)$. Then by the assumption $\min_{l,m}|R_l-R_m|> b$ we have
\begin{equation}\label{myeq2.2}
\mathrm{supp}\,\chi_l\cap\mathrm{supp}\, \chi_m=\emptyset,\ l\neq m.
\end{equation}
Let $B_l$ be the open ball centered at $R_l$ with the radius $b/4$. The left-hand side of \eqref{myeq2.1} is estimated as
\begin{equation}\label{myeq2.3}
\begin{split}
\left\langle f,\sum_{l=1}^L\frac{1}{|x-R_l|}f\right\rangle&\leq \left\langle f,\sum_{m=1}^L\frac{1}{|x-R_m|}\chi_mf\right\rangle+\left\langle f,\sum_{m=1}^L\left(\sum_{l\neq m}\frac{1}{|x-R_l|}\right)\chi_mf\right\rangle\\
&\quad +\int_{\mathbb R^3\setminus\bigsqcup_{l=1}^LB_l}\sum_{l=1}^L\frac{1}{|x-R_l|}|f(x)|^2dx,
\end{split}
\end{equation}
where we used $\mathrm{supp}\, (1-\sum_{m=1}^L\chi_m(x))\subset\mathbb R^3\setminus\bigsqcup_{l=1}^LB_l$ and $0\leq1-\sum_{m=1}^L\chi_m(x)\leq 1$.

The first term in the right-hand side is estimated as
\begin{align*}
\left\langle f,\sum_{m=1}^L\frac{1}{|x-R_m|}\chi_mf\right\rangle&\leq\lVert f\rVert\left\lVert\sum_{m=1}^L\frac{1}{|x-R_m|}\chi_mf\right\rVert\\
&=\lVert f\rVert\left(\sum_{m=1}^L\left\lVert\frac{1}{|x-R_m|}\chi_mf\right\rVert^2\right)^{1/2},
\end{align*}
where in the last equality we used \eqref{myeq2.2}. The Hardy inequality yields
$$\left\lVert\frac{1}{|x-R_m|}\chi_mf\right\rVert^2\leq 4\lVert\nabla(\chi_m f)\rVert^2\leq 8(\lVert(\nabla\chi_m)f\rVert^2+\lVert\chi_m\nabla f\rVert^2),$$
where we used $|(\nabla\chi_m)f+\chi_m\nabla f|^2\leq2|(\nabla\chi_m)f|^2+2|\chi_m\nabla f|^2$ in the last inequality.
Hence we have
\begin{align*}
\sum_{m=1}^L\left\lVert\frac{1}{|x-R_m|}\chi_mf\right\rVert^2&\leq \sum_{m=1}^L8(\lVert(\nabla\chi_m)f\rVert^2+\lVert\chi_m\nabla f\rVert^2)\\
&=8\left(\left\lVert\sum_{m=1}^L(\nabla\chi_m)f\right\rVert^2+\left\lVert\sum_{m=1}^L\chi_m\nabla f\right\rVert^2\right)\\
&=8\left(\frac{64}{b^2}\lVert f\rVert^2+\lVert\nabla f\rVert^2\right),
\end{align*}
where in the second and the third steps we used \eqref{myeq2.2}, $0\leq\chi(r)\leq 1$ and $|\chi'(r)|\leq 8/b$. Thus we obtain
\begin{equation}\label{myeq2.4}
\begin{split}
\left\langle f,\sum_{m=1}^L\frac{1}{|x-R_m|}\chi_mf\right\rangle&\leq\lVert f\rVert \left(\frac{16\sqrt 2}{b}\lVert f\rVert+2\sqrt 2\lVert\nabla f\rVert\right)\\
&\leq(16\sqrt 2b^{-1}+2)\lVert f\rVert^2+\lVert \nabla f\rVert^2.
\end{split}
\end{equation}

In order to estimate the third term in the right-hand side of \eqref{myeq2.3} we estimate the value of the potential $\sum_{l=1}^L\frac{1}{|x-R_l|}$ at $x\in\mathbb R^3\setminus\bigsqcup_{l=1}^LB_l$. First we notice that as is well-known if $x\in B_l^c$, the potential at $x$ given by a charge $1$ at $R_l$ is equal to the electrostatic potential given by the uniform distribution of the charge density $1/(\frac{4}{3}\pi\left(\frac{b}{4}\right)^3)=\frac{48}{b^3\pi}$ in $B_l$, i.e. we have
$$\frac{1}{|x-R_l|}=\frac{48}{b^3\pi}\int_{y\in B_l}\frac{1}{|x-y|}dy.$$
This form of equation is most commonly known in the calculation of the gravitational potential by a star (see e.g. \cite{Mo}), and is the simplest case of the multipole expansion in eletrostatics (see e.g. \cite{Ja}).
Let $\mathcal B_x$ be the open ball centered at $x$ whose volume equals to the Lebesgue measure $\mu(\bigsqcup_{l=1}^LB_l)=\frac{4}{3}\pi\left(\frac{b}{4}\right)^3L$ of the balls. Then the radius of $\mathcal B_x$ is $\frac{b}{4}L^{1/3}$, and we can see that
$$\int_{y\in \bigsqcup_{l=1}^LB_l}\frac{1}{|x-y|}dy\leq\int_{y\in \mathcal B_x}\frac{1}{|x-y|}dy,$$
as follows. We decompose the integral in the left-hand side and estimate as
\begin{align*}
\int_{y\in \bigsqcup_{l=1}^LB_l}\frac{1}{|x-y|}dy&=\int_{y\in (\bigsqcup_{l=1}^LB_l)\cap\mathcal B_x}\frac{1}{|x-y|}dy+\int_{y\in (\bigsqcup_{l=1}^LB_l)\cap\mathcal B_x^c}\frac{1}{|x-y|}dy\\
&\leq\int_{y\in (\bigsqcup_{l=1}^LB_l)\cap\mathcal B_x}\frac{1}{|x-y|}dy+\int_{y\in \mathcal B_x\setminus \bigsqcup_{l=1}^LB_l}\frac{1}{|x-y|}dy\\
&=\int_{y\in \mathcal B_x}\frac{1}{|x-y|}dy,
\end{align*}
where we used that the value of $\frac{1}{|x-y|}$ for $y\in \mathcal B_x$ is greater than that for $y\in \mathcal B_x^c$, and $\mu((\bigsqcup_{l=1}^LB_l)\cap\mathcal B_x^c)=\mu(\mathcal B_x\setminus \bigsqcup_{l=1}^LB_l)$ which follows from the definition of $\mathcal B_x$.

Therefore, we can estimate the potential  at $x\in\mathbb R^3\setminus\bigsqcup_{l=1}^LB_l$ as
\begin{align*}
\sum_{l=1}^L\frac{1}{|x-R_l|}&=\sum_{l=1}^L\frac{48}{b^3\pi}\int_{y\in B_l}\frac{1}{|x-y|}dy=\frac{48}{b^3\pi}\int_{y\in \bigsqcup_{l=1}^LB_l}\frac{1}{|x-y|}dy\\
&\leq\frac{48}{b^3\pi}\int_{y\in \mathcal B_x}\frac{1}{|x-y|}dy=\frac{6L^{2/3}}{b}.
\end{align*}
Thus the third term in the right-hand side of \eqref{myeq2.3} is estimated as
\begin{equation}\label{myeq2.5}
\int_{\mathbb R^3\setminus\bigsqcup_{l=1}^LB_l}\sum_{l=1}^L\frac{1}{|x-R_l|}|f(x)|^2dx\leq\frac{6L^{2/3}}{b}\lVert f\rVert^2_{L^2(\mathbb R^3\setminus\bigsqcup_{l=1}^LB_l)}\leq\frac{6L^{2/3}}{b}\lVert f\rVert^2.
\end{equation}

Since $\mathrm{supp}\, \chi_m\subset B_l^c$ for $l\neq m$, the second term in the right-hand side of \eqref{myeq2.3} is estimated in the same way as the third term and we obtain
\begin{equation}\label{myeq2.6}
\left\langle f,\sum_{m=1}^L\left(\sum_{l\neq m}\frac{1}{|x-R_l|}\right)\chi_mf\right\rangle\leq\frac{6L^{2/3}}{b}\lVert f\rVert^2_{L^2(\bigsqcup_{m=1}^L\mathrm{supp}\, \chi_m)}\leq\frac{6L^{2/3}}{b}\lVert f\rVert^2.
\end{equation}
The result follows from \eqref{myeq2.3}-\eqref{myeq2.6}.
\end{proof}

Using  Lemma \ref{formb} we can obtain a lower bound for $\inf \sigma(h)$ of the spectrum $\sigma(h)$ of $h$.

\begin{lem}\label{oneHb}
Assume that there exists a constant $a>0$ such that $\min_{l,m}|R_l-R_m|>a$. Then we have
$$\inf\sigma(h)\geq-(16\sqrt2+12L^{2/3})\mathcal Z a^{-1}-2\mathcal Z^2,$$
where $\mathcal Z:=\max_{l}Z_l$.
\end{lem}

\begin{proof}
For $s\in\mathbb R$ let us define a unitary operator $U(s)$ of dilation by $(U(s)g)(x):=s^{3/2}g(sx)$. Then we have
$$U(s)^{-1}hU(s)=s^2\left(-\Delta_x-\frac{1}{s}\sum_{l=1}^L\frac{Z_l}{|x-sR_l|}\right).$$
Choosing $s=\mathcal Z$ we have $U(\mathcal Z)^{-1}hU(\mathcal Z)=\mathcal Z^2\tilde h$, where
$$\tilde h:=-\Delta_x-\frac{1}{\mathcal Z}\sum_{l=1}^L\frac{Z_l}{|x-\mathcal ZR_l|}.$$ Applying Lemma \ref{formb} to $f\in H^1(\mathbb R^3)$ and $b=\mathcal Za$ we can estimate the potential as
\begin{align*}
\left\langle f,\frac{1}{\mathcal Z}\sum_{l=1}^L\frac{Z_l}{|x-\mathcal ZR_l|}f\right\rangle&\leq\left\langle f,\sum_{l=1}^L\frac{1}{|x-\mathcal ZR_l|}f\right\rangle\\
&\leq\lVert \nabla f\rVert^2 + \{(16\sqrt2+12L^{2/3})(\mathcal Za)^{-1}+2\}\lVert f\rVert^2.
\end{align*}
Thus we can estimate the form $\langle f,\tilde hf\rangle$ as
\begin{align*}
\langle f,\tilde hf\rangle&\geq\lVert \nabla f\rVert^2-\lVert \nabla f\rVert^2 -\{ (16\sqrt2+12L^{2/3})(\mathcal Za)^{-1}+2\}\lVert f\rVert^2\\
&=- \{(16\sqrt2+12L^{2/3})(\mathcal Za)^{-1}+2\}\lVert f\rVert^2.
\end{align*}
Therefore, we can see that
\begin{align*}
\inf\sigma(h)=\mathcal Z^2\inf\sigma(\tilde h)&\geq-\mathcal Z^2 \{(16\sqrt2+12L^{2/3})(\mathcal Za)^{-1}+2\}
\\
&=-(16\sqrt2+12L^{2/3})\mathcal Z a^{-1}-2\mathcal Z^2,
\end{align*}
which completes the proof.
\end{proof}

\begin{rem}
The assumption $\min_{l,m}|R_l-R_m|>a$ is essential for the lower bound of the order $O(L^{2/3})$ in Lemma \ref{oneHb} as $L$ increases. We can see this as follows. If  $R_l=R_1$ for any $l$, by the unitary operator $U(\sum_{l=1}^LZ_l)$ of dilation we have
$$U(\sum_{l=1}^LZ_l)^{-1}hU(\sum_{l=1}^LZ_l)=(\sum_{l=1}^LZ_l)^2\left(-\Delta_x-\frac{1}{|x-(\sum_{l=1}^LZ_l)R_1|}\right).$$
Hence using that the the infimum of the spectrum of $-\Delta_x-\frac{1}{|x|}$ which corresponds to the first eigenvalue of the hydrogen atom is $-\frac{1}{4}$, we can see that $\inf\sigma (h)=-(\sum_{l=1}^LZ_l)^2/4$. Thus the order of the lower bound is $O(L^2)$. The better order $O(L^{2/3})$ of Lemma \ref{oneHb} is of crucial importance for the right-hand side of \eqref{myeq1.2} to have the order $O(N^{5/6})$ which is smaller than $O(N)$, when $L$ is proportional to $N$.
\end{rem}

\section{relation between one-electron and  two-electron densities}\label{fourthsec}
When we deal with antisymmetric functions in $L^2(\mathbb R^{3N})$, the identification of $L^2(\mathbb R^{3N})$ with $\bigotimes^NL^2(\mathbb R^3)$ is a convenient and standard way (see e.g. \cite[Section II.4]{RS} and \cite[Section 4.5]{Fo}). The identification is given by the linear extension of the mapping $\psi_1\otimes\dotsm\otimes\psi_N\mapsto\psi_1(x_1)\times\dotsm\times\psi_N(x_N)$ for $\psi_i\in L^2(\mathbb R^3)$. By this identification the set of antisymmetric functions is identified with the subspace $\bigwedge^NL^2(\mathbb R^3)\subset\bigotimes^NL^2(\mathbb R^3)$ defined by
$$\bigwedge^NL^2(\mathbb R^3):=\left\{\Phi\in\bigotimes^NL^2(\mathbb R^3):\forall \sigma\in S_N,\ \sigma(\Phi)=(\mathrm{sgn}\, \sigma) \Phi\right\},$$
where $S_N$ is the symmetric group and the action of $\sigma\in S_N$ is defined by the linear extension of $\sigma(\psi_1\otimes\dotsm\otimes\psi_N)=\psi_{\sigma(1)}\otimes\dotsm\otimes\psi_{\sigma(N)}$. The orthogonal projection $A_N$ from $\bigotimes^NL^2(\mathbb R^3)$ onto $\bigwedge^NL^2(\mathbb R^3)$ is given by $A_N:=\frac{1}{N!}\sum_{\sigma\in S_N}(\mathrm{sgn}\, \sigma)\sigma$. We define the exterior product $\psi_1\wedge\dotsm\wedge\psi_N$ of $\psi_1,\dots,\psi_N$ by
$$\psi_1\wedge\dotsm\wedge\psi_N:=\sqrt{N!}A_N(\psi_1\otimes\dotsm\otimes\psi_N)=\frac{1}{\sqrt{N!}}\sum_{\sigma\in S_N}(\mathrm{sgn}\, \sigma)\psi_{\sigma(1)}\otimes\dotsm\otimes\psi_{\sigma(N)}.$$
(Depending on literatures, sometimes exterior product is defined by $\psi_1\wedge\dotsm\wedge\psi_N:=A_N(\psi_1\otimes\dotsm\otimes\psi_N)$, but we choose the definition above in this paper so that $\psi_1\wedge\dotsm\wedge\psi_N$ will be normalized for orthonormal $\psi_i$.)
Then $\psi_1\wedge\dotsm\wedge\psi_N\in\bigwedge^NL^2(\mathbb R^3)$ is identified with the Slater determinant
$$\frac{1}{\sqrt{N!}}\sum_{\sigma\in S_N}(\mathrm{sgn}\, \sigma)\psi_{1}(x_{\sigma(1)})\times\dotsm\times\psi_{N}(x_{\sigma(N)})\in L^2(\mathbb R^{3N}).$$
Note that if $\langle\psi_i,\psi_j\rangle=\delta_{ij},\ 1\leq i, j\leq N$, we have $\lVert \psi_1\wedge\dotsm\wedge\psi_N\rVert=1$. Hereafter, we use the identification of the set of antisymmetric functions and $\bigwedge^NL^2(\mathbb R^3)$ freely. The following lemma about a relation between the one-electron and two-electron densities is of crucial importance in the proof of Theorem \ref{ubd}. Let $P_{\Psi}(\Omega)$ be the probability measure defined by \eqref{myeq1.3} and $\nu_{\Psi}(x,y)$ be the two-electron density \eqref{myeq1.4}.

\begin{lem}\label{onetwod}
Let $\Psi\in L^2(\mathbb R^{3N})$ be a normalized antisymmetric function and $\Omega$ be a Lebesgue measurable set. Then we have
$$P_{\Psi}(\Omega)^2-\frac{1}{N}P_{\Psi}(\Omega)\leq\int_{\substack{x\in\Omega \\ y\in\Omega}}\nu_{\Psi}(x,y)dxdy.$$
\end{lem}
\begin{proof}
\textit{Step 1.} First we shall choose orthonormal bases for $L^2(\mathbb R^3)$ and $\bigwedge^NL^2(\mathbb R^3)$ used in the proof. For this purpose we shall show that for any measurable subset $M\subset\mathbb R^3$, $L^2(M)$ is separable as a Hilbert space. Since $L^2(\mathbb R^3)$ is separable, there exists a countable orthonormal basis $u^j\in L^2(\mathbb R^3),\ j=1,2,\dots$. Any function $f\in L^2(M)$ is extended to a function $\tilde f$ in $L^2(\mathbb R^3)$ by
$$\tilde f(x) = \begin{cases}
f(x) &x\in M\\
0 &x\in M^c
\end{cases}.$$
Then by the completeness of $\{u^j\}$, $\tilde f$ is expanded as $\tilde f=\sum_{j=1}^{\infty}c_ju^j$, $c_j\in\mathbb C$. Then we have $\int_{M}\lvert f(x)-\sum_{j=1}^Jc_ju^j(x)\rvert^2dx\to0$ as $J\to\infty$, Thus the subspace spanned by the restrictions of $u^j$ to $M$ is dense in $L^2(M)$, which implies that $L^2(M)$ is separable.

Therefore, we can see that $L^2(\Omega)$ and $L^2(\Omega^c)$ are separable.  Since separable Hilbert spaces have countable orthonormal bases, we can choose countable orthonormal bases $\{v^{k'}\}$ and $\{w^{k''}\}$ for $L^2(\Omega)$ and $L^2(\Omega^c)$ respectively. We make extensions $\tilde v^{k'}$ and $\tilde w^{k''}$ of $v^{k'}$ and $w^{k''}$ to $L^2(\mathbb R^3)$ by setting $\tilde v^{k'}(x)=0,\ x\in \Omega^c$ and $\tilde w^{k''}(x)=0,\ x\in \Omega$ respectively. Then obviously we have $\langle \tilde v^{k'},\tilde w^{k''}\rangle=0$. Since any function $g$ in $L^2(\mathbb R^3)$ is decomposed as $g=g_1+g_2,\ g_1(x)=0, \forall x\in\Omega^c,\ g_2(x)=0, \forall x\in\Omega$, and $g_1$ and $g_2$ are expanded as $g_1=\sum_{k'=1}^{\infty}a_{k'}\tilde v^{k'},\ a_{k'}\in \mathbb C$ and $g_2=\sum_{k''=1}^{\infty}b_{k''}\tilde w^{k''},\ b_{k''}\in \mathbb C$ respectively, we can see that the union $\{\tilde v^{k'}\}\cup\{\tilde w^{k''}\}$ is an orthonormal basis for $L^2(\mathbb R^3)$. We relabel the basis by the superscript $k\in\mathbb N$ and denote it by $\{\varphi^k\}$. Since either $\varphi^k(x)=0$, $\forall x\in\Omega$ or $\varphi^k(x)=0$, $\forall x\in\Omega^c$ holds and $\varphi^k$ is an orthonormal basis, we have
\begin{equation}\label{myeq3.0}
\int_{\Omega}\overline{\varphi^k(x)}\varphi^{k'}(x)dx=0,\ k\neq k'.
\end{equation}

Using the basis $\{\varphi^k\}$, an orthonormal basis for $\bigwedge^NL^2(\mathbb R^3)$ is given by $\{\varphi^{k_1}\wedge\dotsm\wedge\varphi^{k_N}:k_1<\dotsm<k_N\}$ (see e.g. \cite[Chapter 4]{Fo}). We relabel this basis by $m\in \mathbb N$ and denote it by $\{\Phi^m\}$. Thus $\Phi^m$ can be written as $\Phi^m=\psi^m_1\wedge\dotsm\wedge\psi^m_N$, where $\psi^m_i=\varphi^k$ for some $k$.
\smallskip

\noindent \textit{Step 2.} Let us define a mapping $Q_{\Psi}:\Omega\mapsto\mathbb R$ by the integral
$$Q_{\Psi}(\Omega):=\int_{\substack{x\in\Omega\\ y\in\Omega}}\nu_{\Psi}(x,y)dxdy,$$
of two-electron density of $\Psi$. Let us consider the relation between the one-electron and two-electron densities of a Slater determinant $\Phi^m=\psi_1^m\wedge\dotsm\wedge\psi_N^m$. From the construction of $\Phi^m$ in Step 1 and \eqref{myeq3.0}, we can see that
\begin{equation}\label{myeq3.1}
\int_{\Omega}\overline{\psi_i^m(x)}\psi_j^m(x)dx=\begin{cases}1 &\mathrm{if}\ i=j\ \mathrm{and}\  \psi_i^m(x)=0, \forall x\in\Omega^c\\ 0 &\mathrm{otherwise}\end{cases}.
\end{equation}
Using this property and the orthonormality of $\psi^m_i$ by a direct calculation we have
\begin{align*}
Q_{\Phi^m}(\Omega)&=\int_{\substack{x_1\in\Omega \\ x_2\in\Omega}}\int_{\mathbb R^{3(N-2)}}\lvert\Phi^m(x_1,\dotsm,x_N)\rvert^2 dx_1\dotsm dx_N\\
&=\frac{1}{N!}\sum_{\sigma\in S_N}\int_{\substack{x_1\in\Omega \\ x_2\in\Omega}}\int_{\mathbb R^{3(N-2)}}\lvert\psi^m_1(x_{\sigma(1)})\rvert^2\dotsm\lvert\psi^m_N(x_{\sigma(N)})\rvert^2 dx_1\dotsm dx_N\\
&=\frac{1}{N(N-1)}\sum_{i=1}^{\infty}\sum_{j\neq i}p_i^mp_j^m,
\end{align*}
where $p_i^m:=\int_{\Omega}|\psi^m_i(x)|^2dx$. As for the sum with respect to $j$ we have
\begin{align*}
\left\lvert \frac{1}{N-1}\sum_{j\neq i}p_j^m-\frac{1}{N}\sum_{j=1}^Np_j^m\right\rvert&=\left\lvert \frac{1}{N(N-1)}\left(\sum_{j\neq i}p_j^m-(N-1)p_i^m\right)\right\rvert\\
&=\left\lvert\frac{1}{N}\left(\frac{1}{N-1}\sum_{j\neq i}p_j^m-p_i^m\right)\right\rvert\\
&\leq \frac{1}{N},
\end{align*}
where in the last inequality we used $0\leq p_i^m\leq 1$ and $0\leq \frac{1}{N-1}\sum_{j\neq i}p_j^m\leq 1$. Thus we can see that
\begin{equation}\label{myeq3.2}
Q_{\Phi^m}(\Omega)\geq\left(\frac{1}{N}\sum_{i=1}^Np_i^m\right)^2-\frac{1}{N^2}\sum_{i=1}^Np_i^m.
\end{equation}
On the other hand, by the orthonormality of $\psi^m_i$ and a direct calculation we have
\begin{align*}
P_{\Phi^m}(\Omega)&=\frac{1}{N!}\sum_{\sigma\in S_N}\int_{x_1\in\Omega}\int_{\mathbb R^{3(N-1)}}\lvert\psi^m_1(x_{\sigma(1)})\rvert^2\dotsm\lvert\psi^m_N(x_{\sigma(N)})\rvert^2\\
&=\frac{1}{N}\sum_{i=1}^Np_i^m.
\end{align*}
Combined with \eqref{myeq3.2} this gives
\begin{equation}\label{myeq3.2.1}
Q_{\Phi^m}(\Omega)\geq P_{\Phi^m}(\Omega)^2-\frac{1}{N}P_{\Phi^m}(\Omega).
\end{equation}

\noindent \textit{Step 3.} Since $\{\Phi^m\}$ is an orthonormal basis for $\bigwedge^NL^2(\mathbb R^3)$ we can expand $\Psi$ as $\Psi=\sum_{m=1}^{\infty}d_m\Phi^m$, $d_m\in\mathbb C$, where $\sum_{m=1}^{\infty}|d_m|^2=1$. From the construction of $\Phi^m$ in Step 1, \eqref{myeq3.0} and the orthogonality of $\{\varphi^k\}$ we can see that for $m\neq m'$
\begin{equation}\label{myeq3.3}
\begin{split}
&\int_{x_1\in\Omega}\int_{\mathbb R^{3(N-1)}}\overline{\Phi^m(x_1,\dotsm,x_N)}\Phi^{m'}(x_1,\dotsm,x_N)dx_1\dotsm dx_N=0,\\
&\int_{\substack{x_1\in\Omega\\ x_2\in\Omega}}\int_{\mathbb R^{3(N-2)}}\overline{\Phi^m(x_1,\dotsm,x_N)}\Phi^{m'}(x_1,\dotsm,x_N)dx_1\dotsm dx_N=0.
\end{split}
\end{equation}
It follows from \eqref{myeq3.3} that $P_{\Psi}(\Omega)=\sum_{m=1}^{\infty}|d_m|^2P_{\Phi^m}(\Omega)$ and $Q_{\Psi}(\Omega)= \sum_{m=1}^{\infty}\newline |d_m|^2Q_{\Phi^m}(\Omega)$. Thus by the Cauchy-Schwarz inequality and \eqref{myeq3.2.1} we obtain
\begin{align*}
P_{\Psi}(\Omega)&=\sum_{m=1}^{\infty}|d_m|^2P_{\Phi^m}(\Omega)\\
&\leq \left(\sum_{m=1}^{\infty}|d_m|^2\right)^{1/2}\left(\sum_{m=1}^{\infty}|d_m|^2P_{\Phi^m}(\Omega)^2\right)^{1/2}\\
&=\left(\sum_{m=1}^{\infty}|d_m|^2P_{\Phi^m}(\Omega)^2\right)^{1/2}\\
&\leq\left(\sum_{m=1}^{\infty}|d_m|^2\left(Q_{\Phi^m}(\Omega)+\frac{1}{N}P_{\Phi^m}(\Omega)\right)\right)^{1/2}\\
&=\left(Q_{\Psi}(\Omega)+\frac{1}{N}P_{\Psi}(\Omega)\right)^{1/2}.
\end{align*}
Squaring the both sides and transposing $\frac{1}{N}P_{\Psi}(\Omega)$ we obtain the result.
\end{proof}

\section{proof of Theorem \ref{ubd}}\label{fifthsec}
\begin{proof}
Taking the inner product of $\Psi$ and the both sides of $H\Psi=E\Psi$ we obtain
\begin{align*}
E&=\langle\Psi,H\Psi\rangle\\
&=\sum_{1\leq i<j\leq N}\langle \Psi,\frac{1}{|x_i-x_j|}\Psi\rangle+\sum_{i=1}^N\langle\Psi,h_i\Psi\rangle\\
&\geq \sum_{1\leq i<j\leq N}\int_{\substack{x_i\in\Omega \\ x_j\in \Omega}}\int_{\mathbb R^{3(N-2)}}\frac{1}{|x_i-x_j|}|\Psi(x_1,\dots,x_N)|^2dx_1\dotsm dx_N+\sum_{i=1}^N\langle\Psi,h_i\Psi\rangle\\
&\geq \frac{N(N-1)}{2}\frac{1}{d_{\Omega}}\int_{\substack{x\in\Omega \\ y\in \Omega}}\nu_{\Psi}(x,y)dxdy+\sum_{i=1}^N\langle\Psi,h_i\Psi\rangle\\
&\geq \frac{N(N-1)}{2d_{\Omega}}\left(P_{\Psi}(\Omega)^2-\frac{1}{N}P_{\Psi}(\Omega)\right)-N\{(16\sqrt2+12L^{2/3})\mathcal Za^{-1}+2\mathcal Z^2\},
\end{align*}
where we used the antisymmetry of $\Psi$ in the fourth step and Lemmas \ref{oneHb} and \ref{onetwod} in the last inequality.
Since $E\leq 0$, we have
\begin{align*}
\frac{4d_{\Omega}}{N-1}\{(8\sqrt2+6L^{2/3})\mathcal Za^{-1}+\mathcal Z^2\}&\geq P_{\Psi}(\Omega)^2-\frac{1}{N}P_{\Psi}(\Omega)\\
&=\left(P_{\Psi}(\Omega)-\frac{1}{2N}\right)^2-\frac{1}{4N^2}.
\end{align*}
By transpositions, taking the square root of the both sides, and noting $NP_{\Psi}(\Omega)=\int_{\Omega}\rho_{\Psi}(x)dx$ we obtain the result.
\end{proof}

\def\thesection{\Alph{section}}
\setcounter{section}{0}
\section{Appendix}
In this appendix we prove that $\int_{\Omega}\rho_{\Psi}(x)dx$ is equal to the expectation value $\mathcal E_{\Psi}(\Omega)$ of the number of electrons found in $\Omega$ for a normalized antisymmetric function $\Psi\in L^2(\mathbb R^{3N})$ and a Lebesgue measurable set $\Omega\subset\mathbb R^3$. By definition $\mathcal E_{\Psi}(\Omega)$ is given by
\begin{equation}\label{myeqa.0}
\mathcal E_{\Psi}(\Omega)=\sum_{k=1}^Nk\mathcal P_{\Psi}^k(\Omega),
\end{equation}
where $\mathcal P_{\Psi}^k(\Omega)$ is the probability for the state $\Psi$ to find just $k$ electrons in $\Omega$. The probability $\mathcal P_{\Psi}^k(\Omega)$ is by definition given by
\begin{equation}\label{myeqa.0.1}
\begin{split}
\mathcal P_{\Psi}^k(\Omega)&=\frac{1}{k!(N-k)!}\sum_{\tau\in S_N}\\
&\quad\cdot\int_{x_{\tau(1)}\in \Omega}\dotsm\int_{x_{\tau(k)}\in \Omega}\int_{x_{\tau(k+1)}\in \Omega^c}\dotsm\int_{x_{\tau(N)}\in \Omega^c}|\Psi(x_1,\dots,x_N)|^2dx_1\dotsm dx_N\\
&=\frac{N!}{k!(N-k)!}\\
&\quad\cdot\int_{x_{1}\in \Omega}\dotsm\int_{x_{k}\in \Omega}\int_{x_{k+1}\in \Omega^c}\dotsm\int_{x_{N}\in \Omega^c}|\Psi(x_1,\dots,x_N)|^2dx_1\dotsm dx_N,
\end{split}
\end{equation}
where in the first equality we used that the same combination of variables of integration in $\Omega$ is counted $k!$ times and that in $\Omega^c$ is counted $(N-k)!$ times by the permutations $\tau$. The second equality follows from the antisymmetry of $\Psi$.

Let $\{\Phi^m\}$ be the basis set for $\bigwedge^NL^2(\mathbb R^3)$ constructed in the proof of Lemma \ref{onetwod}. Let us first prove $\int_{\Omega}\rho_{\Phi^m}(x)dx=\mathcal E_{\Phi^m}(\Omega)$. Recall that each function $\psi^m_i\in L^2(\mathbb R^3)$ in $\Phi^m=\psi^m_1\wedge\dotsm\wedge\psi^m_N$ satisfies either $\psi_i^m(x)=0,\ \forall \Omega$ or $\psi_i^m(x)=0,\ \forall \Omega^c$. Therefore, as for the density of $\Phi^m$ using the orthonormality of $\psi^m_i$ we have
$$\int_{\Omega}\rho_{\Phi^m}(x)dx=\int_{\Omega}(|\psi^m_1(x)|^2+\dotsm+|\psi^m_N(x)|^2)dx=n,$$
where $n=n^m$ is the number of $\psi^m_i$ which has nonzero values only in $\Omega$. Without loss of generality we may assume that the first $n$ functions $\psi^m_1,\dots,\psi^m_{n}$ have nonzero values only in $\Omega$ and the other functions have nonzero values only in $\Omega^c$. Then by \eqref{myeqa.0} and \eqref{myeqa.0.1} the expectation value $\mathcal E_{\Phi^m}(\Omega)$ is given by
\begin{align*}
&\mathcal E_{\Phi^m}(\Omega)=\sum_{k=1}^Nk\mathcal P_{\Phi^m}^k(\Omega)\\
&\quad=\sum_{k=1}^Nk\binom{N}{k}\int_{x_1\in\Omega}\dotsm\int_{x_k\in\Omega}\int_{x_{k+1}\in\Omega^c}\dotsm\int_{x_N\in\Omega^c}|\Phi^m(x_1,\dots,x_N)|^2dx_1\dotsm dx_N\\
&\quad=n\binom{N}{n}\frac{1}{N!}\sum_{\sigma\in S_n}\int_{x_1\in\Omega}\dotsm\int_{x_{n}\in\Omega}\int_{x_{n+1}\in\Omega^c}\dotsm\int_{x_N\in\Omega^c}|\psi^m_1(x_{\sigma(1)})|^2\\
&\hspace{230pt}\dotsm|\psi^m_N(x_{\sigma(N)})|^2dx_1\dotsm dx_N\\
&\quad=n\binom{N}{n}\frac{1}{N!}n!(N-n)!\\
&\qquad\qquad\int_{x_1\in\Omega}\dotsm\int_{x_{n}\in\Omega}\int_{x_{n+1}\in\Omega^c}\dotsm\int_{x_N\in\Omega^c}|\psi^m_1(x_1)|^2\dotsm|\psi^m_N(x_1)|^2dx_1\dotsm dx_N\\
&\quad=n,
\end{align*}
where in the third equality we used that for $k\neq n$ the integral vanishes, \eqref{myeq3.1} and
$$\int_{\Omega^c}\overline{\psi_i^m(x)}\psi_j^m(x)dx=\begin{cases}1 &\mathrm{if}\ i=j\ \mathrm{and}\  \psi_i^m(x)=0, \forall x\in\Omega\\ 0 &\mathrm{otherwise}\end{cases}.$$
In the fourth equality we used that the number of the permutations of $\{1,\dots,n\}$ is $n!$ and that of the rest indices is $(N-n)!$. Thus we obtain $\int_{\Omega}\rho_{\Phi^m}(x)dx=\mathcal E_{\Phi^m}(\Omega)$.

For general normalized antisymmetric $\Psi$ we expand $\Psi$ by $\Phi^m$ as $\Psi=\sum_{m=1}^{\infty}c_m\Phi^m$. Here we note that as in \eqref{myeq3.3} using \eqref{myeq3.0} and
$$\int_{\Omega^c}\overline{\varphi^k(x)}\varphi^{k'}(x)dx=0,\ k\neq k',$$
we have for $m\neq m'$ 
\begin{equation}\label{myeqa.1}
\begin{split}
&\int_{x_1\in\Omega}\dotsm\int_{x_k\in\Omega}\int_{x_{k+1}\in\Omega^c}\dotsm\int_{x_N\in\Omega^c}\overline{\Phi^m(x_1,\dots,x_N)}\Phi^{m'}(x_1,\dots,x_N)dx_1\dotsm dx_N\\
&\quad =0.
\end{split}
\end{equation}
It follows from \eqref{myeq3.3} and \eqref{myeqa.1} that $\int_{\Omega}\rho_{\Psi}(x)dx=\sum_{m=1}^{\infty}|c_m|^2\int_{\Omega}\rho_{\Phi^m}(x)dx$ and $\mathcal E_{\Psi}(\Omega)=\sum_{m=1}^{\infty}|c_m|^2\mathcal E_{\Phi^m}(\Omega)$. Hence $\int_{\Omega}\rho_{\Phi^m}(x)dx=\mathcal E_{\Phi^m}(\Omega)$ yields $\int_{\Omega}\rho_{\Psi}(x)dx=\mathcal E_{\Psi}(\Omega)$, which completes the proof.

\end{document}